\documentclass[11pt]{article}
\usepackage[T2A]{fontenc}
\usepackage{ alphalph,etoolbox }
\usepackage{amsthm, amsmath, amssymb, amsbsy, amscd, amsfonts, latexsym, euscript,epsf, bm}
\newtheorem{thm}{Theorem}[section]
\usepackage{bbold, bm, stackrel}
\usepackage{epic,eepic}
\usepackage{texdraw}
\usepackage[dvips]{color}
\usepackage{color}
\usepackage{ dsfont }
\usepackage{graphicx}
\usepackage{exscale,relsize}
\usepackage{authblk}
\usepackage{url}
\usepackage{hyperref}
\usepackage{enumerate}
\usepackage{graphicx}
\usepackage{multicol}
\newtheorem{proposition}[thm]{Proposition}
\newtheorem{corollary}[thm]{Corollary}

\newtheorem{theorem}[thm]{Theorem}

\theoremstyle{definition}

\newtheorem{example}[thm]{Example}

\DeclareMathOperator{\End}{End}
\DeclareMathOperator{\id}{id}

\newcommand{\qqquad}{\qquad\quad}

\title{On the solutions of the local Zamolodchikov  tetrahedron equation}
\date{}

\author{M. Chirkov\thanks{mikhlchirkov@gmail.com} ~and S. Konstantinou-Rizos\thanks{skonstantin84@gmail.com}}
\affil{Centre of Integrable Systems, P.G. Demidov Yaroslavl State University, Yaroslavl, Russia}

\patchcmd{\subequations}{\alph{equation}}{\alphalph{\value{equation}}}{}{}

\begin{document}

\maketitle

\begin{abstract}
We study the solutions of the local Zamolodhcikov tetrahedron equation in the form of correspondences derived by $3\times 3$ matrices. We present all the associated generators of $4$-simplex maps satisfying the local tetrahedron equation. Moreover, we demonstrate that, from some of our solutions, we can recover the  4-simplex extensions of Kashaev--Korepanov--Sergeev and Hirota type tetrahedron maps. Finally, we construct several novel 4-simplex maps.
\end{abstract}

\bigskip


\begin{quotation}
\end{quotation}

\begin{quotation}
\end{quotation}

\begin{quotation}
\end{quotation}

\section{Introduction}\label{intro}
The $n$-simplex equations are a family of important  equations of Mathematical Physics. The first member of the family is trivial, while the second, third and fourth members are quite famous  and  have attracted the  interest of many leading scientists from various fields of Mathematics and Physics. 

In particular,  the  2-simplex  equation is the Yang--Baxter equation \cite{Drinfeld} which is one of the most  fundamental equations of  Mathematical Physics, the 3-simplex equation is the  well-celebrated Zamolodchikov's tetrahedron equation \cite{Zamolodchikov, Zamolodchikov-2}, and the 4-simplex equation is  the Bazhanov--Stroganov equation \cite{Bazhanov}.

Due to the importance of $n$-simplex maps, there is a need to develop methods for constructing and classifying them. Famous results towards the classification of $n$-simplex maps, include the classification of Yang--Baxter maps by Adler, Bobenko and Suris \cite{ABS-2004}, as well as the classification of Zamolodchikov tetrahedron maps by Kashaev, Korepanov and Sergeev \cite{Kashaev-Sergeev, Sergeev}.  There are several methods in the literature for constructing $n$-simplex maps; we indicatively refer to the derivation of $n$-simplex maps via the symmetries of lattice equations \cite{Pap-Tongas-Veselov, Kassotakis-Tetrahedron}, their construction via matrix refactorisation problems and their construction via Darboux transformations for integrable PDEs \cite{Sokor-Sasha, Sokor-Papamikos, Papamikos}.

The most significant property that share all the $n$-simplex equations is that the local $(n-1)$-simplex equation is a generator of solutions to the $n$-simplex equation \cite{Dimakis-Hoissen}. It is known that any unique solution of the local $(n-1)$-simplex equation satisfies the $n$-simplex equation. However, there is no proof of such a statement; thus, in order to prove that a map satisfying uniquely the local $(n-1)$-simplex is an $n$-simplex map, we either need to substitute it to  the  $n$-simplex equation, or use  supplementary matrix conditions \cite{Sokor-2022, Kouloukas}.

Recently, it has become understood that non-unique solutions  (in the form of correspondences) of the local $(n-1)$-simplex equation also give rise to $n$-simplex maps \cite{Doliwa-Kashaev, Igonin-Sokor, Sokor-2020, Sokor-2023, Sokor-Kouloukas}. Moreover, this correspondence approach gives rise to novel $n$-simplex maps that do not belong to  any of the known  classification lists. Most known classification results focus on  $n$-simplex maps themselves. It makes sense to classify the solutions of the local $(n-1)$-simplex equation which generate $n$-simplex maps which  may  lead  to  wider classification results.

In this paper, we study the solutions of the local Zamolodchikov tetrahedron equation. Specifically, we:
\begin{itemize}
    \item Present all the possible solutions to the local tetrahedron equation in the form of correspondences generated by $3\times 3$ matrices depending on free variables;
    \item We demonstrate that these solutions generate all known in  the literature 4-simplex maps. As illustrative examples we use Kashaev--Sergeev--Korepanov and Hirota type 4-simplex maps \cite{Sokor-2023, Sokor-2024};
    \item We use the obtained local tetrahedron correspondences to construct new, noninvolutive, 4-simplex maps.
\end{itemize}

The rest of the text is organised as follows: In the next section, we  provide all the necessary preliminary information which makes the text self-contained. In particular, we fix the notation that is used throughout the text, we explain what 4-simplex maps are, and what is their relation to the local tetrahedron equation. Section \ref{sol-local-tetra} deals with the study of all possible solutions of the local Zamolodchikov tetrahedron equation generated by $3\times 3$ matrices. We  show that the essentially different solutions to the local tetrahedron equation in the form of correspondences are generated by four simple $3\times 3$ matrices. These correspondences are generators of 4-simplex maps. Finally, in section \ref{4simplexMaps}, we construct novel, noninvolutive, 4-simplex maps.

\section{Preliminaries}
\subsection{Notation}
Throughout the text:
\begin{itemize}
    \item By $\mathcal{X}$ we denote an arbitrary set, whereas by Latin italic letters (i.e. $x, y, u, v$ etc.) the elements of $\mathcal{X}$. 

    \item  By $\mathcal{X}^n$ we denote the Cartesian product: $\mathcal{X}^n=\underbrace{\mathcal{X}\times \ldots \times \mathcal{X}}_\text{$n$}$.
  
    \item Matrices will be denoted by capital Latin straight letters (i.e. ${\rm A}, {\rm B}, {\rm C}$) etc. Additionally, by ${\rm A}^n_{i_1 i_2\ldots i_m}$ we denote the $n\times n$ extensions of the $m\times m$ ($m<n$) matrix ${\rm A}$, where the elements of ${\rm A}$ are placed at the intersection of the $i_1, i_2\ldots i_m$ rows with the $i_1, i_2\ldots i_m$ columns of matrix ${\rm A}^n_{i_1 i_2\ldots i_m}$. For example, ${\rm A}^3_{12}=\begin{pmatrix} 
 a_{11} &  a_{12} & 0\\ 
a_{21} &  a_{22} & 0\\
0 & 0 & 1
\end{pmatrix}$, ${\rm A}^3_{13}=\begin{pmatrix} 
 a_{11} & 0 &  a_{12} \\
 0 & 1 & 0\\
a_{21} & 0 & a_{22} 
\end{pmatrix}$ and  ${\rm A}^3_{23}=\begin{pmatrix} 
1 & 0 & 0\\
0 & a_{11} &  a_{12} \\ 
0 & a_{21} &  a_{22} \\
\end{pmatrix}$
are $3\times 3$ extensions of matrix ${\rm A}=\begin{pmatrix} 
 a_{11} &  a_{12} \\ 
a_{21} &  a_{22} \\
\end{pmatrix}$.
\end{itemize}

\subsection{Local Zamolodchikov tetrahedron  equation  and 4-simplex maps}
As it  was mentioned earlier, the local $3$-simplex equation is a generator of solutions of the $4$-simplex equation. In this section, we  give the definitions of the local Zamolodchikov tetrahedron equation and a $4$-simplex map.

Let $\mathcal{X}$ be a set. A map $S:\mathcal{X}^4\rightarrow \mathcal{X}^4$, namely
\begin{equation}\label{4-simplex-map}
 S:(x,y,z,t)\mapsto (u(x,y,z,t),v(x,y,z,t),w(x,y,z,t),r(x,y,z,t)),
\end{equation}
is called a \textit{4-simplex map} if it satisfies the \textit{set-theoretical 4-simplex}  (or Bazhanov--Stroganov) equation
\begin{equation}\label{4-simplex-eq}
    S^{1234}\circ S^{1567} \circ S^{2589}\circ S^{368,10} \circ S^{479,10}=S^{479,10}\circ S^{368,10}\circ S^{2589}\circ S^{1567} \circ S^{1234}.
\end{equation}
Functions $S^{ijkl}\in\End(\mathcal{X}^{10})$, $i,j,k,l=1,\ldots 10,~i< j<k<l$, in \eqref{4-simplex-eq} are maps that act as map $S$ on the $ijkl$ terms of the Cartesian product $\mathcal{X}^{10}$ and trivially on the others. For instance,
\begin{align*}
S^{1567}&(x_1,x_2,\ldots x_{10})=\\
&(u(x_1,x_5,x_6,x_7),x_2,x_3,x_4,v(x_1,x_5,x_6,x_7),w(x_1,x_5,x_6,x_7),r(x_1,x_5,x_6,x_7),x_8,x_9,x_{10}).
\end{align*}

Now, let ${\rm L}={\rm L}(x)$ be a $3\times 3$ square matrix depending on a variable $x\in\mathcal{X}$ and a parameter $\kappa\in\mathbb{C}$ of the form
\begin{equation}\label{matrix-L-simplex}
   {\rm L}(x)= \begin{pmatrix} 
a(x) & b(x) & c(x)\\ 
d(x) & e(x) & f(x)\\
k(x) & l(x) & m(x)
\end{pmatrix},
\end{equation}
where its entries $a, b, c, d, e, f, k, l$ and $m$ are scalar functions of $x$ and $k$. Let ${\rm L}^6_{ijk}(x;\kappa)$, $i,j,k=1,\ldots 6$, $i< j<k$, be the $6\times 6$ extensions of matrix \eqref{matrix-L-simplex},  defined by \allowdisplaybreaks
{\small
\begin{subequations}\label{6x6-extensions}
\begin{align}
    & {\rm L}^6_{123}(x)= \begin{pmatrix} 
a(x) & b(x) & c(x) & 0 & 0 & 0\\ 
d(x) & e(x) & f(x) & 0 & 0 & 0\\
k(x) & l(x) & m(x) & 0 & 0 & 0\\
0 & 0 & 0 & 1 & 0 & 0 \\
0 & 0 & 0 & 0 & 1 & 0 \\
0 & 0 & 0 & 0 & 0 & 1 \\
\end{pmatrix}, \quad
{\rm L}^6_{145}(x)= \begin{pmatrix} 
a(x) & 0 & 0 & b(x) & c(x) & 0\\ 
0 & 1 & 0 & 0 & 0 & 0 \\
0 & 0 & 1 & 0 & 0 & 0 \\
d(x) & 0 & 0 & e(x) & f(x) & 0\\
k(x) & 0 & 0 & l(x) & m(x) & 0\\
0 & 0 & 0 & 0 & 0 & 1 \\
\end{pmatrix},\\
& {\rm L}^6_{246}(x)= \begin{pmatrix} 
1 & 0 & 0 & 0 & 0 & 0 \\
0 & a(x) & 0 & b(x) & 0 & c(x)\\ 
0 & 0 & 1 & 0 & 0 & 0 \\
0 & d(x) & 0 & e(x) & 0 & f(x)\\
0 & 0 & 0 & 0 & 1 & 0 \\
0 & k(x) & 0 & l(x) & 0 & m(x)
\end{pmatrix},\quad 
 {\rm L}^6_{356}(x)= \begin{pmatrix} 
 1 & 0 & 0 & 0 & 0 & 0 \\
 0 & 1 & 0 & 0 & 0 & 0 \\
0 & 0 & a(x) & 0 & b(x) & c(x)\\ 
0 & 0 & 0 & 1 & 0 & 0 \\
0 & 0 & d(x) & 0 & e(x) & f(x)\\
0 & 0 & k(x) & 0 & l(x) & m(x)
\end{pmatrix}.
\end{align}
\end{subequations}
}

We call the following matrix four-factorisation problem
\begin{equation}\label{local-tetrahedron}
    {\rm L}^6_{123}(u){\rm L}^6_{145}(v){\rm L}^6_{246}(w){\rm L}^6_{356}(r)={\rm L}^6_{356}(t){\rm L}^6_{246}(z){\rm L}^6_{145}(y){\rm L}^6_{123}(x)
\end{equation}
\textit{local $3$-simplex equation} or \textit{local Zamolodchikov tetrahedron equation}. The local tetrahedron equation is a generator of potential solutions to the 4-simplex equation. If map \eqref{4-simplex-map} satisfies equation \eqref{local-tetrahedron}, then the matrix refactorisation problem \eqref{local-tetrahedron} is called a \textit{Lax representation} for map \eqref{4-simplex-map}.

Examples of $4$-simplex maps can  be found in \cite{Sokor-2023, Hietarinta, Bardakov}.

\section{Solutions of the local tetrahedron  equation}\label{sol-local-tetra}
In this section, we aim to study the solutions  to the local tetrahedron equation in the form of correspondences generated by $3\times  3$ matrices. Specifically, let $x_i\in \mathbb{C}$, $i=1,\ldots, 9$, be free variables  and
\begin{equation}\label{matrix-3x3}
{\rm L}(x_1,\ldots, x_9)=\begin{pmatrix}
    x_1 & x_2 & x_3\\
    x_4 & x_5 & x_6\\
    x_7 & x_8 & x_9
\end{pmatrix}.
\end{equation}
Then, substitute \eqref{matrix-3x3} to the local tetrahedron  equation:
\begin{align}\label{3x3-Tetr}
    & {\rm L}^6_{123}(u_1,\ldots, u_9){\rm L}^6_{145}(v_1,\ldots, v_9){\rm L}^6_{246}(w_1,\ldots, w_9){\rm L}^6_{356}(r_1,\ldots, r_9)=\nonumber\\&{\rm L}^6_{356}(t_1,\ldots, t_9){\rm L}^6_{246}(z_1,\ldots, z_9){\rm L}^6_{145}(y_1,\ldots, y_9){\rm L}^6_{123}(x_1,\ldots, x_9).
\end{align}

We seek matrices \eqref{matrix-3x3} for which the matrix refactorisation problem \eqref{3x3-Tetr} has solution for $u_i$, $v_i$, $w_i$  and  $r_i$, $i=1,\ldots, 9$. In particular  we have  the following.

\begin{theorem}
    Let $x_i\in \mathbb{C}$, $i=1,\ldots, 9$, be free variables. Then, the essentially different solutions of the local Zamolodchikov tetrahedron equation \eqref{local-tetrahedron} are given  in the form of correspondences generated by matrices of the form:
    \begin{equation}\label{matrix-forms}
      \rm{A}=  \begin{pmatrix}
    x_1 & x_2 & 0\\
    x_3 & x_4 & 0\\
    0 & 0 & x_5
\end{pmatrix},\quad
 \rm{B}=  \begin{pmatrix}
    x_1 & x_2 & 0\\
    x_3 & x_4 & x_5\\
    0 & 0 & 0
\end{pmatrix},\quad
 \rm{C}= \begin{pmatrix}
    0 & x_1 & 0\\
    x_2 & 0 & x_3\\
    0 & x_4 & 0
\end{pmatrix},\quad
 \rm{D}= \begin{pmatrix}
    0 & 0 & 0\\
    x_1 & 0 & 0\\
    0 & x_2 & x_3
\end{pmatrix},
    \end{equation}
where ${\rm  A}=  {\rm  A} (x_1,x_2,x_3,x_4,x_5)$, ${\rm  B}=  {\rm  B} (x_1,x_2,x_3,x_4,x_5)$, $ {\rm  C}=  {\rm  C}(x_1,x_2,x_3,x_4)$ and $ {\rm  D}=  {\rm  D} (x_1,x_2,x_3)$.
    All the rest of solutions are generated by matrices \eqref{matrix-forms}, after:
    \begin{enumerate}
        \item Reflection of $x_i$'s with respect to the antidiagonal;
        \item Rearranging $x_i$'s among themselves;
        \item Setting any variable $x_i=0$.
    \end{enumerate}
\end{theorem}
\begin{proof}
    We substitute matrix ${\rm  A}$ to the local tetrahedron equation: 
    \begin{align*}
        &{\rm A}^6_{123}(u_1,\ldots u_5){\rm A}^6_{145}(v_1,\ldots v_5){\rm A}^6_{246}(w_1,\ldots w_5){\rm A}^6_{356}(r_1,\ldots r_5)=\\
        &{\rm A}^6_{356}(t_1,\ldots t_5){\rm A}^6_{246}(z_1,\ldots z_5){\rm A}^6_{145}(y_1,\ldots y_5){\rm A}^6_{123}(x_1,\ldots x_5).
    \end{align*}
    The above matrix tetra-factorisation problem is equivalent to the system of polynomial equations
   \begin{eqnarray}
       &u_1v_1=x_1y_1,\quad u_2w_1+u_1v_2w_3=x_2y_1\quad u_2w_2+u_1v_2w_4=y_2, \quad u_3v_1=x_3z_1+x_1y_3z_2,&\nonumber\\
       &u_4 w_1 + u_3 v_2 w_3 = x_4 z_1 + x_2 y_3 z_2,\quad u_4 w_2 + u_3 v_2 w_4 = y_4 z_2,\quad r_1 u_5 = t_1,x_5,\quad r_2 u_5 = t_2 y_5,&\label{corr-A}\\
       &v_3 = x_3 z_3 +x_1 y_3 z_4,\quad v_4 w_3 = x_4 z_3 + x_2 y_3 z_4,\quad v_4 w_4 = y_4 z_4,&\nonumber\\
&r_3 v_5 = t_3 x_5,\quad r_4 v_5 = t_4 y_5,\quad r_5 w_5 = t_5 z_5.&\nonumber
   \end{eqnarray}
The above system has infinitely many solutions for $u_i$, $v_i$, $w_i$ and $r_i$, $i=1,\ldots,5$. For instance, it can be solved for $u_1$, $u_3$, $u_4$, $v_1$, $v_4$, $w_1$, $w_2$, $w_4$, $r_1$, $r_2$, $r_3$, $r_4$, $r_5$ in terms of $u_5$, $v_2$, $v_3$, $v_5$, $w_3$ and $w_5$. Thus, system \eqref{corr-A} defines a correspondence between $\mathbb{C}^{20}$ and $\mathbb{C}^{20}$. However, we do not write this correspondence explicitly because of its length.

Now, we substitute matrix ${\rm  B}$ to the local tetrahedron equation: 
    \begin{align*}
        &{\rm B}^6_{123}(u_1,\ldots u_5){\rm B}^6_{145}(v_1,\ldots v_5){\rm B}^6_{246}(w_1,\ldots w_5){\rm B}^6_{356}(r_1,\ldots r_5)=\\
        &{\rm B}^6_{356}(t_1,\ldots t_5){\rm B}^6_{246}(z_1,\ldots z_5){\rm B}^6_{145}(y_1,\ldots y_5){\rm B}^6_{123}(x_1,\ldots x_5).
    \end{align*}
It follows that the above matrix tetra-factorisation problem is equivalent to the system of polynomial equations
\begin{eqnarray}
       &u_1 v_1 = x_1 y_1,\quad u_2 w_1 + u_1 v_2 w_3 = x_2 y_1,\quad u_2 w_2 + u_1 v_2 w_4 = y_2,\quad u_3 v_1 = x_3 z_1 + x_1 y_3 z_2,&\nonumber\\
       &u_4 w_1 + u_3 v_2 w_3 = x_4 z_1 + x_2 y_3 z_2,\quad  u_4 w_2 + u_3 v_2 w_4 = y_4 z_2,\quad r_1 u_5 = x_5 z_1,\quad\quad r_2 u_5 = y_5 z_2&\label{corr-B}\\
       &v_3 = x_3 z_3 + x_1 y_3 z_4,\quad v_4 w_3 = x_4 z_3 + x_2 y_3 z_4,\quad v_4 w_4 = y_4 z_4,&\nonumber\\
&r_3 v_5 = x_5 z_3,\quad r_4 v_5 = y_5 z_4,\quad r_5 v_5 = z_5.&\nonumber
   \end{eqnarray}
   The above system has infinitely many solutions for $u_i$, $v_i$, $w_i$ and $r_i$, $i=1,\ldots,5$. For example, it can be solved for $u_1$, $u_3$, $u_4$, $v_1$, $v_4$, $w_1$, $w_2$, $w_4$, $r_1$, $r_2$, $r_3$, $r_4$, $r_5$ in terms of the  rest variables $u_5$, $v_2$, $v_3$, $v_5$, $w_3$ and $w_5$. Thus, system \eqref{corr-B} defines a correspondence between $\mathbb{C}^{20}$ and $\mathbb{C}^{20}$. We omit this correspondence because of its length.

   Next, we substitute matrix ${\rm  C}$ to the local tetrahedron equation: 
    \begin{align*}
        &{\rm C}^6_{123}(u_1,\ldots u_4){\rm C}^6_{145}(v_1,\ldots v_4){\rm C}^6_{246}(w_1,\ldots w_4){\rm C}^6_{356}(r_1,\ldots r_4)=\\
        &{\rm C}^6_{356}(t_1,\ldots t_4){\rm C}^6_{246}(z_1,\ldots z_4){\rm C}^6_{145}(y_1,\ldots y_4){\rm C}^6_{123}(x_1,\ldots x_4).
    \end{align*}
It turns out that the above matrix tetra-factorisation problem is equivalent to the system of polynomial equations
\begin{eqnarray}
       &u_1 w_1 = y_1,\quad u_2 v_1 w_2 = x_1 y_2 z_1,\quad r_1 u_3 + r_4 u_2 v_1 w_3 = y_3 z_1,\quad u_4 w_1 =t_1 y_4,\quad v_2 = x_2 z_2,&\nonumber\\
       & r_2 v_3 = x_3 z_2,\quad r_3 v_3 = z_3,\quad v_4 w_2 = t_2 x_4 + t_3 x_1 y_2 z_4,\quad r_4 v_4 w_3  =t_3 y_3 z_4,\quad  w_4 = t_4 y_4.&\label{corr-C}
   \end{eqnarray}
    The above system has infinitely many solutions for $u_i$, $v_i$, $w_i$ and $r_i$, $i=1,\ldots,4$. For instance, it can be solved for $u_1$, $u_4$, $v_2$, $v_4$, $w_2$, $w_4$, $r_1$, $r_2$, $r_3$, $r_4$, in terms of the  rest variables $u_2$, $u_3$, $v_1$, $v_3$, $w_1$ and $w_3$, namely
    \begin{align*}
        &u_1=\frac{y_1}{w_1},\quad u_4= \frac{t_1 y_4}{w_1},\quad v_2 = x_2 z_2,\quad v_4 = \frac{u_2 v_1(t_2 x_4 + t_3 x_1 y_2 z_4)}{x_1 y_2 z_1},\quad w_2 = \frac{x_1 y_2 z_1}{u_2 v_1},\quad w_4 = t_4 y_4,\\
        &r_1 = \frac{t_2 x_4 y_3 z_1}{u_3(t_2 x_4 + t_3 x_1 y_2 z_4)},\quad r_2 = \frac{x_3 z_2}{v_3},\quad r_3 = \frac{z_3}{v_3},\quad r_4 = \frac{t_3 x_1 y_2 y_3 z_1 z_4}{u_2 v_1 w_3(t_2  x_4 + t_3 x_1 y_2 z_4)}.
    \end{align*}
    Hence, system \eqref{corr-C} defines the above correspondence between $\mathbb{C}^{16}$ and $\mathbb{C}^{16}$. 
    
    Finally, we substitute matrix ${\rm  D}$ to the local tetrahedron equation: 
    \begin{align*}
        &{\rm D}^6_{123}(u_1,u_2, u_3){\rm D}^6_{145}(v_1,v_2, v_3){\rm D}^6_{246}(w_1, w_2, w_3){\rm D}^6_{356}(r_1, r_2,r_3)=\\
        &{\rm D}^6_{356}(t_1, t_2, t_3){\rm D}^6_{246}(z_1,z_2, z_3){\rm D}^6_{145}(y_1, y_2,y_3){\rm D}^6_{123}(x_1,x_2, x_3).
    \end{align*}
The above matrix tetra-factorisation problem is equivalent to the system of polynomial equations
\begin{equation}
       v_1 = x_1 z_1,\quad v_2 w_1 = t_1 x_2,\quad
 r_1 v_3 = t_1 x_3,\quad w_2 = t_2 y_2,\quad r_2 w_3 = t_2 y_3,\quad r_3 w_3 = t_3 z_3.\label{corr-D}
   \end{equation}
    The above system has infinitely many solutions for $u_i$, $v_i$, $w_i$ and $r_i$, $i=1,2,3$. For instance, it can be solved for $u_1$, $u_4$, $v_2$, $v_4$, $w_2$, $w_4$, $r_1$, $r_2$, $r_3$, $r_4$, in terms of the  rest variables $u_2$, $u_3$, $v_1$, $v_3$, $w_1$ and $w_3$, namely
    \begin{equation*}
        v_1 = x_1 z_1,\quad w_1 = \frac{t_1 x_2}{v_2},\quad w_2 = t_2 y_2,\quad r_1 = \frac{t_1 x_3}{v_3},\quad r_2 = \frac{t_2 y_3}{w_3},\quad r_3 = \frac{t_3 z_3}{w_3}.
    \end{equation*}
    That is, system \eqref{corr-D} defines the above correspondence between $\mathbb{C}^{12}$ and $\mathbb{C}^{12}$. 

It can be readily verified that, setting any $x_i=0$ in one of the matrices in \eqref{matrix-forms}, the resulting polynomial equations can be solved for the remaining nonzero $u_i$, $v_i$, $w_i$ and $r_i$.

Furthermore, it can be verified that the reflection of the elements with respect to the antidiagonal in any of the matrices in \eqref{matrix-forms} will leave systems \eqref{corr-A}, \eqref{corr-B}, \eqref{corr-C} and \eqref{corr-D} covariant. Specifically, changing ${\rm A}$, ${\rm B}$, ${\rm C}$ and ${\rm D}$ to
$$ 
{\rm \tilde{A}}={\rm P} {\rm A}^T {\rm P},\quad {\rm \tilde{B}}={\rm P} {\rm B}^T {\rm P},\quad {\rm \tilde{C}}={\rm P} {\rm C}^T {\rm P},\quad
{\rm \tilde{D}}={\rm P} {\rm D}^T {\rm P},\quad \text{where}  ~~ {\rm P}=\begin{pmatrix} 0 & 0 & 1\\ 0 & 1 & 0\\ 1 & 0 & 0\end{pmatrix},
$$ 
and substituting the latter to the local tetrahedron equation, we shall obtain exactly the same correspondences \eqref{corr-A}, \eqref{corr-B}, \eqref{corr-C} and \eqref{corr-D}, respectively, after changing in them $(u,v,x,y)\longleftrightarrow (r,u,t,z)$. In Figure 1 we demonstrate this fact in a particular example.

Finally, after replacing any of the zero elements in \eqref{matrix-forms} with a free variable, one can easily verify that  the local tetrahedron equation has no solutions for $u_i$, $v_i$, $w_i$ and $r_i$.
\end{proof}
\hspace{1cm}
\begin{figure}[ht]
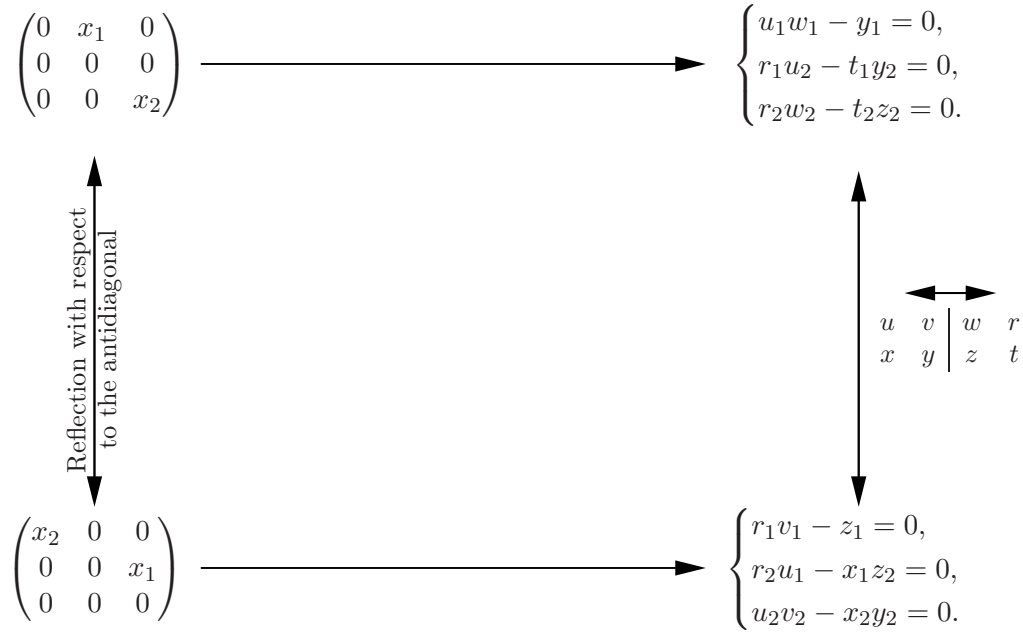

\centering
\centertexdraw{ 
\setunitscale 0.8
\move (-1.3 -.2)  \arrowheadtype t:F \avec(2 -.2)
\textref h:C v:C \htext(-2 -.2){ $\begin{pmatrix}
0 & x_1 & 0 \\ 0 &  0 & 0 \\ 0 & 0 &  x_2
\end{pmatrix}$}
\textref h:C v:C \htext(3 -.2){ $\begin{cases}
u_1w_1-y_1=0,\\
r_1u_2-t_1y_2=0,\\
r_2w_2-t_2z_2=0.
\end{cases}$}
\move (3 -1)  \arrowheadtype t:F \avec(3 -3.1)
\move (3 -3)  \arrowheadtype t:F \avec(3 -.9)
 \move (3.4 -1.7)  \arrowheadtype t:F \avec(3.9 -1.7)
\move (3.8 -1.7) \arrowheadtype t:F \avec(3.3 -1.7)
\textref h:C v:C \htext(3.6 -2){\small
$\begin{array}{c c|c c}
u  & v & w & r\\
x & y & z & t
\end{array}$}
\textref h:C v:C \htext(3 -3.5){$\begin{cases}
r_1v_1-z_1=0,\\
r_2u_1-x_1z_2=0,\\
u_2v_2-x_2y_2=0.
\end{cases}$}
\move (-1.3 -3.5)  \arrowheadtype t:F \avec(2 -3.5)
\textref h:C v:C \htext(-2 -3.5){$\begin{pmatrix}
x_2 & 0 & 0 \\ 0 &  0 & x_1 \\ 0 & 0 & 0
\end{pmatrix}$}
\move (-2 -1)  \arrowheadtype t:F \avec(-2 -3.1)
\move (-2 -3)  \arrowheadtype t:F \avec(-2 -.8)
\textref h:C v:C \small{\vtext(-2.1 -2){\small Reflection with respect}}
\textref h:C v:C \small{\vtext(-1.9 -2){\small  to the antidiagonal}}
}\vspace{1cm}
\caption{Reflection symmetry example.}\label{Mirroring}
\end{figure}

\begin{corollary}
    The maximum dimension of a map generated by substitution of matrix \eqref{matrix-3x3} to the local Zamolodchikov tetrahedron equation is 20.
\end{corollary}

\section{Set-theoretical solutions to the 4-simplex equation}\label{4simplexMaps}
In this section, we use the solutions of the local Zamolodchikov tetrahedron equation in order to construct solutions to the  Bazhanov--Stroganov 4-simplex equation.

The first, probably, 4-simplex maps in the literature appear in the work of Hietarinta \cite{Hietarinta}, and also in \cite{Bardakov}. The first 4-simplex maps generated by the local Zamolodchikov tetrahedron equation appeared in \cite{Sokor-2023, Sokor-2024}. In particular, in \cite{Sokor-2023, Sokor-2024} 4-simplex maps were constructed as extensions of well-known tetrahedron maps, using matrices of the form ${\rm A}$ in \eqref{matrix-forms}. Thus, all 4-simplex maps which appear in \cite{Sokor-2023, Sokor-2024} can be generated by system \eqref{corr-A} after a change of variables.

\subsection{Recovery of known 4-simplex maps}
Here, we present some examples of known 4-simplex maps which  can be recovered by the more general solutions to the local tetrahedron equation of the previous section. 

\begin{example} (Kashaev--Sergeev--Korepanov type 4-simplex maps) Changing in \eqref{corr-A} 
\begin{align*}
    (x_1,y_1,z_1,t_1,u_1,v_1,w_1,r_1)&\rightarrow (1,1,1,1,1,1,1,1),\\
    (x_2,y_2,z_2,t_2,u_2,v_2,w_2,r_2)&\rightarrow (x_1,y_1,z_1,t_1,u_1,v_1,w_1,r_1),\\
    (x_3,y_3,z_3,t_3,u_3,v_3,w_3,r_3)&\rightarrow  \left(\frac{k}{x_1},\frac{k}{y_1},\frac{k}{z_1},\frac{k}{t_1},\frac{k}{u_1},\frac{k}{v_1},\frac{k}{w_1},\frac{k}{r_1}\right),\\
    (x_4,y_4,z_4,t_4,u_4,v_4,w_4,r_4)&\rightarrow  (0,0,0,0,0,0,0,0),\\
    (x_5,y_5,z_5,t_5,u_5,v_5,w_5,r_5)&\rightarrow (x_2,y_2,z_2,t_2,u_2,v_2,w_2,r_2),
\end{align*}
and solving the associated equations for $u_i$, $v_i$, $w_i$ and $r_i$, $i=1,2$, we obtain the following correspondence between $\mathbb{C}^{8}$ and $\mathbb{C}^{8}$:
 $$ 
 u_1=\frac{x_1y_1}{y_1+x_1z_1},\quad u_2=x_2,\quad v_1=\frac{x_1z_1}{k},\quad v_2=y_2,\quad w_1=\frac{y_1+x_1z_1}{x_1},\quad r_1=\frac{t_1y_2}{x_2},\quad r_2=\frac{t_2z_2}{w_2}.
 $$

Fixing $w_2=t_2$ or $w_2=y_2$, we recover maps
\begin{align*}
     S_1: (x_1,x_2,y_1,y_2,z_1,z_2,t_1,t_2)&\rightarrow \left(\frac{x_1y_1}{y_1+x_1z_1}, x_2, \frac{x_1z_1}{k}, y_2, \frac{y_1+x_1z_1}{x_1}, t_2, \frac{t_1y_2}{x_2}, z_2\right),\\
   S_2: (x_1,x_2,y_1,y_2,z_1,z_2,t_1,t_2)&\rightarrow \left(\frac{x_1y_1}{y_1+x_1z_1}, x_2, \frac{x_1z_1}{k}, y_2, \frac{y_1+x_1z_1}{x_1}, y_2, \frac{t_1y_2}{x_2}, \frac{t_2z_2}{y_2}\right),
\end{align*}
respectively, from \cite{Sokor-2023}.
 \end{example}

 \begin{example} (Hirota type 4-simplex maps) Changing in \eqref{corr-A} $(x_2,y_2,z_2,t_2,u_2,v_2,w_2,r_2)\rightarrow  (1,1,1,1,1,1,1,1)$, and then
\begin{align*}
    (x_1,y_1,z_1,t_1,u_1,v_1,w_1,r_1)&\rightarrow \left(\frac{x_1}{x_2},\frac{y_1}{y_2},\frac{z_1}{z_2},\frac{t_1}{t_2},\frac{u_1}{u_2},\frac{v_1}{v_2},\frac{w_1}{w_2},\frac{r_1}{r_2}\right),\\
    (x_3,y_3,z_3,t_3,u_3,v_3,w_3,r_3)&\rightarrow  (1,1,1,1,1,1,1,1),\\
    (x_4,y_4,z_4,t_4,u_4,v_4,w_4,r_4)&\rightarrow  (0,0,0,0,0,0,0,0),\\
    (x_5,y_5,z_5,t_5,u_5,v_5,w_5,r_5)&\rightarrow (x_2,y_2,z_2,t_2,u_2,v_2,w_2,r_2),
\end{align*}
and solving the associated equations for $u_i$, $v_i$, $w_i$ and $r_i$, $i=1,2$, we obtain the following correspondence between $\mathbb{C}^{8}$ and $\mathbb{C}^{8}$:
\begin{equation*}
  u_1=\frac{x_1y_1z_2}{x_2z_1+x_1z_2},~~ u_2=y_2,~~ v_1=\frac{x_2z_1+x_1z_2}{z_2},~~ v_2=x_2,~~ w_1=\frac{w_2x_2y_1z_1}{y_2(x_2z_1+x_1z_2)},~~ r_1=\frac{t_1x_2z_2}{w_2y_2},~~ r_2=\frac{t_2z_2}{w_2}.
\end{equation*}

For the choices of the free variable $w_2=t_2$ and $w_2=z_2$ the above correspondence defines the following Hirota type 4-simplex maps:
\begin{align*}
      S_1: (x_1,x_2,y_1,y_2,z_1,z_2,t_1,t_2)&\rightarrow \left(\frac{x_1y_1z_2}{x_2z_1+x_1z_2},y_2,\frac{x_2z_1+x_1z_2}{z_2},x_2,\frac{t_2x_2y_1z_1}{y_2(x_2z_1+x_1z_2)},t_2,\frac{t_1x_2z_2}{t_2y_2},z_2\right),\\
   S_2: (x_1,x_2,y_1,y_2,z_1,z_2,t_1,t_2)&\rightarrow \left(\frac{x_1y_1z_2}{x_2z_1+x_1z_2},y_2,\frac{x_2z_1+x_1z_2}{z_2},x_2,\frac{x_2y_1z_1z_2}{y_2(x_2z_1+x_1z_2)},z_2,\frac{t_1x_2}{y_2},t_2\right),
\end{align*}
which appeared in \cite{Sokor-2024}.
 \end{example} 

The above are illustrative examples, and after certain change of variables, one can recover all the 4-simplex maps which appeared in \cite{Sokor-2023, Sokor-2024}. Moreover, we can construct more general 4-simplex maps which can be restricted to the 4-simplex maps presented in \cite{Sokor-2023, Sokor-2024} at a particular limit. At the same time, as we shall see in the next section, we can construct new hierarchies of 4-simplex maps, other than the ones generated by matrix ${\rm A}$ in \eqref{matrix-forms}, by employing matrices ${\rm B}$, ${\rm C}$ and ${\rm D}$ in \eqref{matrix-forms}.

\subsection{Construction of novel 4-simplex maps}
In section \ref{sol-local-tetra}, we studied the solutions of the local Zamolodchikov tetrahedron equation that are obtained by $3\times 3$ matrices. In this section, we employ the systems of polynomial equations \eqref{corr-A}, \eqref{corr-B}, \eqref{corr-C} and \eqref{corr-D}, for the construction of novel 4-simplex maps. We will adopt the correspondence  approach \cite{Igonin-Sokor}.

We start with correspondence \eqref{corr-A}. In \eqref{corr-A} there are 14 independent equations. Indeed, for the system of equations
$$
P_i(u_1,\ldots u_5,v_1,\ldots v_5,w_1,\ldots w_5,r_1,\ldots r_5,x_1,\ldots x_5,y_1,\ldots y_5,z_1,\ldots z_5,t_1,\ldots t_5)=0,~~ i=1,\ldots,14,
$$
where
\begin{align*}
       &P_1=u_1v_1-x_1y_1,~~ P_2=u_2w_1+u_1v_2w_3-x_2y_1,~~ P_3=u_2w_2+u_1v_2w_4-y_2,~~ P_4=u_3v_1-x_3z_1-x_1y_3z_2,\\
       &P_5=u_4 w_1 + u_3 v_2 w_3 - x_4 z_1 - x_2 y_3 z_2,~~ P_6=u_4 w_2 + u_3 v_2 w_4 - y_4 z_2,~~ P_7=r_1 u_5 - t_1,x_5,\\
       &P_8=r_2 u_5 - t_2 y_5,~~ P_9=v_3 - x_3 z_3 -x_1 y_3 z_4,~~ P_{10}=v_4 w_3 - x_4 z_3 - x_2 y_3 z_4,~~ P_{11}=v_4 w_4 - y_4 z_4,\\
&P_{12}=r_3 v_5 - t_3 x_5,~~ P_{13}= r_4 v_5 - t_4 y_5,P_{14}=r_5 w_5 = t_5 z_5,
   \end{align*}
we have that the $14\times 20$ matrix $\left[Q_1\ldots Q_{14} \right]$, consisting of vector rows $Q_i=\nabla P_i$, $i=1,\dots,14$ has rank 14. That is, in order to define a map, we need 6 ($=20-14$) more equations. 

In what follows, we shall see that for particular supplementary equations, the correspondence \eqref{corr-A} implies a novel 4-simplex map.

\begin{proposition}
    The map $S_A:\mathbb{C}^{20}\rightarrow \mathbb{C}^{20}$, 
    \begin{equation}\label{mapS-A}(x_1,\ldots,x_5,y_1,\ldots,y_5,z_1,\ldots,z_5,t_1,\ldots,t_5)\stackrel{S_A}{\longrightarrow}(u_1,\ldots,u_5,v_1,\ldots,v_5,w_1,\ldots,w_5,r_1,\ldots,r_5),
    \end{equation}
    where
    \begin{align*}
        x_1&\mapsto u_1=x_1y_1,\\
        x_2&\mapsto u_2=\frac{(x_2 x_3 - x_1 x_4) \left[x_4 y_2 z_3 + x_2 (y_2 y_3 - y_1 y_4) z_4\right]}{(x_3 y_2 z_3 + x_1 y_1 y_4 z_4) x_4 z_1 + (x_2 x_3 z_1 z_4 + x_1 x_2 y_3 z_2 z_4 + 
    x_1 x_4 z_2 z_3)  (y_2 y_3 - y_1 y_4)},\\
     x_3&\mapsto u_3 = x_3 z_1 + x_1 y_3 z_2,\\
      x_4&\mapsto u_4=\frac{x_4 y_4 (x_1 x_4 - x_2 x_3)  (z_1 z_4 - z_2 z_3)}{(x_3 y_2 z_3 + x_1 y_1 y_4 z_4) x_4 z_1 + (x_2 x_3 z_1 z_4 + x_1 x_2 y_3 z_2 z_4 + 
    x_1 x_4 z_2 z_3)  (y_2 y_3 - y_1 y_4)},\\
    x_5&\mapsto u_5 = x_5,\\
    y_1&\mapsto v_1 = 1,\\
    y_2&\mapsto v_2 = \frac{x_4 y_2 z_1 + 
  x_2 (y_2 y_3 - y_1 y_4) z_2}{(x_1 x_4 - x_2 x_3) (z_1 z_4 - z_2 z_3)},\\
   y_3&\mapsto v_3 = x_3 z_3 + x_1 y_3 z_4,\\
   y_4&\mapsto v_4 =\frac{(x_3 y_2 z_3 + x_1 y_1 y_4 z_4) x_4 z_1 + (x_2 x_3 z_1 z_4 + x_1 x_2 y_3 z_2 z_4 + 
    x_1 x_4 z_2 z_3)  (y_2 y_3 - y_1 y_4)}{(x_1 x_4 - x_2 x_3) (z_1 z_4 - z_2 z_3)},\\
     y_5&\mapsto v_5 =y_5,\\
     z_1&\mapsto w_1=y_1z_1,\\
z_2&\mapsto w_2 = \frac{ x_1 z_2 ( y_1 y_4 -y_2 y_3) -x_3 y_2 z_1}{x_1 x_4 - x_2 x_3 },\\
z_3&\mapsto w_3 =\frac{(x_1 x_4-x_2 x_3) (x_4 z_3 + x_2 y_3 z_4) (z_1 z_4-z_2 z_3)}{(x_3 y_2 z_3 + x_1 y_1 y_4 z_4) x_4 z_1 + (x_2 x_3 z_1 z_4 + x_1 x_2 y_3 z_2 z_4 + 
    x_1 x_4 z_2 z_3)  (y_2 y_3 - y_1 y_4)},\\
    z_4&\mapsto w_4 =\frac{(x_1 x_4 - x_2 x_3 ) y_4 z_4 (z_1 z_4 -z_2 z_3)}{(x_3 y_2 z_3 + x_1 y_1 y_4 z_4) x_4 z_1 + (x_2 x_3 z_1 z_4 + x_1 x_2 y_3 z_2 z_4 + 
    x_1 x_4 z_2 z_3)  (y_2 y_3 - y_1 y_4)},\\
    z_5&\mapsto w_5 =  z_5,\\
    t_1&\mapsto r_1 = t_1,\\
    t_2&\mapsto r_2 = \frac{t_2y_5}{x_5},\\
    t_3&\mapsto r_3 = \frac{t_3x_5}{y_5},\\
    t_4&\mapsto r_4 = t_4,\\
    t_5&\mapsto r_5 = t_5,
    \end{align*}
    is a noninvolutive 4-simplex map and it admits the following functionally independent invariants:
    \begin{eqnarray}
&&I_1=x_1y_1,\quad I_2=y_1z_1,\quad I_3=x_5,\quad I_4=y_5,\quad I_5=z_5,\\
&&I_6=t_1,\quad I_7=t_5,\quad I_8=t_2t_3,\quad I_9=y_4z_4,\quad I_{10}=x_4y_4.
    \end{eqnarray}
\end{proposition}
\begin{proof}
    Map $S_A$ can be obtained as a unique solution for $u_i$, $v_i$, $w_i$ and $r_i$ of the system of equations \eqref{corr-A} supplemented by the following equations:
    \begin{eqnarray*}
        &u_1 = x_1 y_1,\quad u_5=x_5,\quad v_5 = y_5,\quad w_5 = z_5,&\\
        &u_5(u_1 u_4-u_2 u_3) = x_5 (x_1 x_4-x_2 x_3),\quad v_5(v_1 v_4-v_2 v_3) = y_5(y_1 y_4-y_2 y_3).& 
    \end{eqnarray*}

    The fact that $S_A$ is a 4-simplex map can be verified by straightforward substitution of \eqref{mapS-A} to the 4-simplex equation \eqref{4-simplex-eq}. Moreover, we have, for instance, that $r_2\circ S_A=\frac{t_2y_5^2}{x_5^2}\neq t_2$, namely $_AS^2\neq \id$, which means that $S$ is noninvolutive.

    Finally, it can be checked that $I_i\circ S_A=I_i$, $i=1,\ldots,10$, and that the vectors $\nabla I_i$,  $i=1,\ldots,10$, are linearly independent, which means that $I_i$, $i=1,\ldots,10$, are functionally independent invariants of map $S_A$.
\end{proof}

Now, it is easy to show that \eqref{corr-B} has 14 independent equations. In order to define a map, we need 6 supplementary equations. In particular, we have the following.

\begin{proposition}
    The map $S_B:\mathbb{C}^{20}\rightarrow \mathbb{C}^{20}$, 
     \begin{equation}\label{mapS-B}(x_1,\ldots,x_5,y_1,\ldots,y_5,z_1,\ldots,z_5,t_1,\ldots,t_5)\stackrel{S_B}{\longrightarrow}(u_1,\ldots,u_5,v_1,\ldots,v_5,w_1,\ldots,w_5,r_1,\ldots,r_5),
    \end{equation}
    where
    \begin{align*}
        x_1&\mapsto u_1=x_1y_1,\\
        x_2&\mapsto u_2=\frac{(x_2 x_3 - x_1 x_4) \left[x_4 y_2 z_3 + x_2 (y_2 y_3 - y_1 y_4) z_4\right]}{(x_3 y_2 z_3 + x_1 y_1 y_4 z_4) x_4 z_1 + (x_2 x_3 z_1 z_4 + x_1 x_2 y_3 z_2 z_4 + 
    x_1 x_4 z_2 z_3)  (y_2 y_3 - y_1 y_4)},\\
     x_3&\mapsto u_3 = x_3 z_1 + x_1 y_3 z_2,\\
      x_4&\mapsto u_4=\frac{x_4 y_4 (x_1 x_4 - x_2 x_3)  (z_1 z_4 - z_2 z_3)}{(x_3 y_2 z_3 + x_1 y_1 y_4 z_4) x_4 z_1 + (x_2 x_3 z_1 z_4 + x_1 x_2 y_3 z_2 z_4 + 
    x_1 x_4 z_2 z_3)  (y_2 y_3 - y_1 y_4)},\\
    x_5&\mapsto u_5 = x_5,\\
    y_1&\mapsto v_1 = 1,\\
    y_2&\mapsto v_2 = \frac{x_4 y_2 z_1 + 
  x_2 (y_2 y_3 - y_1 y_4) z_2}{(x_1 x_4 - x_2 x_3) (z_1 z_4 - z_2 z_3)},\\
   y_3&\mapsto v_3 = x_3 z_3 + x_1 y_3 z_4,\\
   y_4&\mapsto v_4 =\frac{(x_3 y_2 z_3 + x_1 y_1 y_4 z_4) x_4 z_1 + (x_2 x_3 z_1 z_4 + x_1 x_2 y_3 z_2 z_4 + 
    x_1 x_4 z_2 z_3)  (y_2 y_3 - y_1 y_4)}{(x_1 x_4 - x_2 x_3) (z_1 z_4 - z_2 z_3)},\\
     y_5&\mapsto v_5 =y_5,\\
     z_1&\mapsto w_1=y_1z_1,\\
z_2&\mapsto w_2 = \frac{ x_1 z_2 ( y_1 y_4 -y_2 y_3) -x_3 y_2 z_1}{x_1 x_4 - x_2 x_3 },\\
z_3&\mapsto w_3 =\frac{(x_1 x_4-x_2 x_3) (x_4 z_3 + x_2 y_3 z_4) (z_1 z_4-z_2 z_3)}{(x_3 y_2 z_3 + x_1 y_1 y_4 z_4) x_4 z_1 + (x_2 x_3 z_1 z_4 + x_1 x_2 y_3 z_2 z_4 + 
    x_1 x_4 z_2 z_3)  (y_2 y_3 - y_1 y_4)},\\
    z_4&\mapsto w_4 =\frac{(x_1 x_4 - x_2 x_3 ) y_4 z_4 (z_1 z_4 -z_2 z_3)}{(x_3 y_2 z_3 + x_1 y_1 y_4 z_4) x_4 z_1 + (x_2 x_3 z_1 z_4 + x_1 x_2 y_3 z_2 z_4 + 
    x_1 x_4 z_2 z_3)  (y_2 y_3 - y_1 y_4)},\\
    z_5&\mapsto w_5 =  y_5t_5,\\
    t_1&\mapsto r_1 = z_1,\\
    t_2&\mapsto r_2 = \frac{y_5z_2}{x_5},\\
    t_3&\mapsto r_3 = \frac{x_5z_3}{y_5},\\
    t_4&\mapsto r_4 =z_4,\\
    t_5&\mapsto r_5 = \frac{z_5}{y_5},
    \end{align*}
    is a noninvolutive 4-simplex map and it admits the following functionally independent invariants:
    \begin{eqnarray*}
&&I_1=x_1y_1,\quad I_2=y_1z_1,\quad I_3=x_5,\quad I_4=y_5,\quad I_5=y_4z_4,\quad I_6=x_4y_4\\
&&I_7=t_5y_5z_5,\quad I_8= x_1 x_4-x_2 x_3,\quad I_9= y_1 y_4-y_2 y_3,\quad I_{10}= z_1 z_4-z_2 z_3.
    \end{eqnarray*}
\end{proposition}
\begin{proof}
    Map $S_B$ can be obtained as a unique solution for $u_i$, $v_i$, $w_i$ and $r_i$ of the system of equations \eqref{corr-B} supplemented by the following equations:
    \begin{eqnarray*}
        &u_1 = x_1 y_1,\quad u_5=x_5,\quad v_5 = y_5,\quad w_5r_5 = z_5t_5,&\\
        &u_5(u_1 u_4-u_2 u_3) = x_5 (x_1 x_4-x_2 x_3),\quad v_5(v_1 v_4-v_2 v_3) = y_5(y_1 y_4-y_2 y_3).& 
    \end{eqnarray*}

    It can be verified by straightforward substitution that $S_B$ satisfies the 4-simplex equation \eqref{4-simplex-eq}. Additionally, we have, for instance, that $w_1\circ S_B=y_1z_1\neq y_1$, namely $S_B^2\neq \id$, i.e. $S_B$ is noninvolutive.

    Finally, $I_i\circ S_B=I_i$, $i=1,\ldots,10$, which means that $I_i$, $i=1,\ldots,10$,  are invariants of $S_B$. Also, vectors $\nabla I_i$,  $i=1,\ldots,10$, are linearly independent, that is $I_i$, $i=1,\ldots,10$, are functionally independent.
\end{proof}

Next, consider system \eqref{corr-C}. It is easy to verify that the latter has 10 independent equations. In order to define a map, we need 6 supplementary equations. Specifically, we have the following.

\begin{proposition}
    The map $S_C:\mathbb{C}^{16}\rightarrow \mathbb{C}^{16}$, 
      \begin{equation}\label{mapS-C}(x_1,\ldots,x_4,y_1,\ldots,y_4,z_1,\ldots,z_4,t_1,\ldots,t_4)\stackrel{S_C}{\longrightarrow}(u_1,\ldots,u_4,v_1,\ldots,v_4,w_1,\ldots,w_4,r_1,\ldots,r_4),
    \end{equation}
\begin{align*}
 x_1&\mapsto u_1=\frac{y_1}{t_1},\qqquad\qqquad\quad~~ z_1\mapsto w_1=t_1,\\
        x_2&\mapsto u_2=y_2,\qqquad\qqquad\qquad z_2\mapsto w_2 =1,\\
     x_3&\mapsto u_3 =y_3,\qqquad\qqquad\qquad z_3\mapsto w_3 =t_3y_3,\\
      x_4&\mapsto u_4=y_4,\qqquad\qqquad\qquad z_4\mapsto w_4 =t_4y_4,\\
    y_1&\mapsto v_1 =x_1z_1,\qqquad\qqquad\quad~ t_1\mapsto r_1 = \frac{t_2 x_4 z_1}{t_2 x_4 + t_3 x_1 y_2 z_4},\\
    y_2&\mapsto v_2 =x_2z_2,\qqquad\qqquad\quad~ t_2\mapsto r_2 =z_2,\\\\
   y_3&\mapsto v_3 =x_3,\qqquad\qqquad\qquad~ t_3\mapsto r_3 =\frac{z_3}{x_3},\\
   y_4&\mapsto v_4 =t_2 x_4 + t_3 x_1 y_2 z_4,\qqquad  t_4\mapsto r_4 =\frac{z_4}{t_2 x_4 + t_3 x_1 y_2 z_4},
    \end{align*}
    is a noninvolutive 4-simplex map and it admits the following functionally independent invariants:
    \begin{eqnarray*}
&I_1=x_3y_3,\quad I_2=x_2y_2z_2,\quad I_3=y_3z_3t_3,\quad I_4=x_1y_1z_1,&\\
&I_5=x_3+y_3,\quad I_6=y_1+x_1z_1,\quad I_7=z_3+t_3y_3,\quad I_8=y_2+x_2z_2.&
    \end{eqnarray*}
\end{proposition}
\begin{proof}
    Supplementing the system of equations \eqref{corr-C} by the following equations:
    \begin{equation*}
        u_2 v_1 = x_1 y_2 z_1,\quad v_3 = x_3,\quad w_3 = t_3 y_3,\quad u_3 = y_3,\quad w_1 = t_1,\quad u_2 = y_2, 
    \end{equation*}
    we obtain a unique solution for $u_i$, $v_i$, $w_i$ and $r_i$, $i=1,2,3,4$, which defines map \eqref{mapS-C}.

    Substitution that $S_C$ to the 4-simplex equation \eqref{4-simplex-eq} confirms  that map \eqref{mapS-C} is a 4-simplex map. Moreover, we have, for instance, that $w_1\circ S_C=\frac{t_2 x_4 z_1}{t_2 x_4 + t_3 x_1 y_2 z_4}\neq z_1$, namely $S_C^2\neq \id$, i.e. $S_C$ is noninvolutive.

    Finally, $I_i\circ S_C=I_i$, $i=1,\ldots,8$, which means that $I_i$, $i=1,\ldots,8$,  are invariants of $S_C$. Furthermore, vectors $\nabla I_i$,  $i=1,\ldots,8$, are linearly independent, which implies that $I_i$, $i=1,\ldots,8$, are functionally independent.
\end{proof}

Lastly, we consider system \eqref{corr-D}. It can be verified that the \eqref{corr-D} has 6 independent equations. In order to define a map, we need 6 supplementary equations. In particular, we have the following.

\begin{proposition}
    The map $S_D:\mathbb{C}^{12}\rightarrow \mathbb{C}^{12}$, 
    \begin{equation}\label{mapS-D}(x_1,x_2,x_3,y_1,y_2,y_3,z_1,z_2,z_3,t_1,t_2,t_3)\stackrel{S_C}{\longrightarrow}\left(x_1,x_2,x_3,x_1z_1,y_2,y_3,\frac{t_1x_2}{y_2},t_2y_2,z_3,\frac{t_1x_3}{y_3},\frac{t_2y_3}{z_3},t_3\right),
    \end{equation}
    is a noninvolutive 4-simplex map and it admits the following functionally independent invariants:
    $$
    I_1=x_1,\quad I_2=x_2,\quad I_3=x_3,\quad I_4=y_2,\quad I_5=y_3,\quad I_6=z_3.
    $$
\end{proposition}
\begin{proof}
    Supplementing system \eqref{corr-D} by the following equations:
    \begin{equation*}
        u_1 = x_1,\quad u_2 = x_2,\quad u_3 = x_3,\quad v_2 = y_2,\quad v_3 = y_3,\quad w_3 = z_3,
    \end{equation*}
    we obtain a unique solution for $u_i$, $v_i$, $w_i$ and $r_i$, $i=1,2,3$, which defines map \eqref{mapS-D}.

    After substitution of $S_D$ to the 4-simplex equation \eqref{4-simplex-eq} we obtain  that map \eqref{mapS-D} is a 4-simplex map. Additionally, we have, for instance, that $r_1\circ S_D=\frac{t_1 x_3^2}{y_3^2}\neq t_1$, namely $S_D^2\neq \id$, thus $S_D$ is noninvolutive.

    The invariants of  the map are obvious and so is  their functional independence .
\end{proof}

\section{Conclusions}
In this paper we studied the set-theoretical solutions to the local Zamolodchikov tetrahedron equation that generate set-theoretical solutions to the Bazhanov--Stroganov 4-simplex equation. In particular, we studied the solutions of the local tetrahedron equation in the form of correspondences between $\mathbb{C}^n$ and $\mathbb{C}^n$, generated by $3\times 3$ matrices. Although, usually, in the literature $n$-simplex maps appear as unique solutions to the $(n-1)$-simplex equation, we adopt the correspondence approach \cite{Igonin-Sokor} which gives rise to much wider hierarchies of solutions.

Specifically, we found that all the essentially different solutions to the Zamolodchikov tetrahedron equation generated by $3\times 3$ matrices are  given by matrices ${\rm A, B, C}$ and ${\rm D}$ in  \eqref{matrix-forms}. In fact, these matrices generate all the possible solutions to the local 4-simplex equation in the form of correspondences \eqref{corr-A}, \eqref{corr-B}, \eqref{corr-C} and \eqref{corr-D}. We demonstrated that \eqref{corr-A} generates all known 4-simplex maps, using Kashaev--Korepanov--Sergeev and Hirota type 4-simplex maps as illustrative examples. Then, we employed  the correspondences \eqref{corr-A}, \eqref{corr-B}, \eqref{corr-C} and \eqref{corr-D} to construct new 4-simplex maps, namely maps \eqref{mapS-A}, \eqref{mapS-B}, \eqref{mapS-C} and \eqref{mapS-D}. All these maps are noninvolutive which, in terms of dynamics, are more interesting that involutive maps which have trivial dynamics.

Our results can be extended in the following ways:
\begin{itemize}
    \item Study the Liouville integrability of the 4-simplex maps presented in this paper.

    We demonstrated that all the 4-simplex maps constructed in this paper admit enough functionally independent first  integrals; that is, half the dimension of the map in number. This  already  indicates their integrability, however their complete (Liouville) integrability is  an open problem.
    
    \item Employ the correspondences \eqref{corr-A}, \eqref{corr-B}, \eqref{corr-C} and \eqref{corr-D} to construct parametric 4-simplex maps.

    We focused on non-parametric 4-simplex maps in this paper, however, supplementing systems \eqref{corr-A}, \eqref{corr-B}, \eqref{corr-C} and \eqref{corr-D} by equations containing parameters, one may construct new solutions to the parametric 4-simplex equation. Several parametric 4-simplex maps were presented in \cite{Sokor-2023}.
    
    \item Classify the solutions of the local tetrahedron equation.

    Our motivation to study the solutions of the local Zamolodchikov tetrahedron equation is that such  solutions generate 4-simplex maps. This, together with the fact that correspondences satisfying the local tetrahedron equation may give generate infinitely many solutions (and also give rise to solutions that  do not belong to the known classification lists \cite{Igonin-Sokor}), suggests that it makes sense to classify the solutions of the local tetrahedron equation. The classification of  the latter solutions is an open problem.

    In general, we propose to classify the generators of $n$-simplex maps, namely the solutions to the local $n$-simplex equation. This will give rise to wider from the already known classification lists even in the case of Yang--Baxter (2-simplex) maps.
    
    \item Study the solutions of the local pentagon equation.

    The $n$-gon equations have equally rich algebro-geometric structure as the $n$-simplex equations \cite{Dimakis-Hoissen, Doliwa-Sergeev, Hoissen}. The most famous member of the family of $n$-gon equations is the $5$-gon or pentagon equation (see, e.g. \cite{Dimakis-Hoissen, Doliwa-Sergeev, Kassotakis-2023, Hoissen} and the references therein). Similar results to those obtained in this paper can be obtained for the local pentagon equation. Again, the motivation comes from the fact that the local $(n-1)$-gon equation is a generator of solutions to the $n$-gon equation.
    
    \item Extend the  results to the noncommutative case.

    Noncommutative $n$-simplex maps have been in the centre of interest for scientists from the mathematical physics community (see, e.g. \cite{Giota-Miky, Doliwa-Kashaev, Kassotakis-Kouloukas, Sokor-2022} and the references therein), due to the rapid development of noncommutative integrable systems. One could extend our  results by considering that the elements of matrix \eqref{matrix-3x3} are matrices, or, more generally,  free variables from a noncommutative division ring.
\end{itemize}

\section*{Acknowledgements}
This work was funded by the Russian Science Foundation project No. 20-71-10110 (https://rscf.ru/en/project/23-71-50012/).
Part of this  work,  namely the proof of Theorem 3.1, was carried out within the framework of a development programme for the Regional Scientific and Educational Mathematical Centre of the P.G. Demidov Yaroslavl State University with financial support from the Ministry of Science and Higher Education of the Russian Federation (Agreement on provision of subsidy from the federal budget No. 075-02-2023-948).

We would like to thank S. Igonin for several useful discussions.

\end{document}